\documentclass[letterpaper, 10 pt, conference]{ieeeconf}  

\IEEEoverridecommandlockouts                              
\usepackage{graphicx} 
\usepackage{amsmath} 
\usepackage{amssymb}  
\usepackage{xcolor}
\usepackage{cite}
\newtheorem{theorem}{Theorem}
\newtheorem{proposition}[theorem]{Proposition}
\newtheorem{lemma}[theorem]{Lemma}

\title{\LARGE \bf
Distributionally Robust Tuning of Anomaly Detectors \\in Cyber-Physical Systems with Stealthy Attacks**
}

\author{Venkatraman Renganathan*, Navid Hashemi*, Justin Ruths, and Tyler H. Summers
\thanks{* These authors contributed equally.}
 \thanks{**This work was partially supported by the Army Research Office and was accomplished under Grant Number: W911NF-17-1-0058 and the Air Force Office of Scientific Research under award number FA2386-19-1-4073.}
\thanks{The authors are with the Department of Mechanical Engineering, The University of Texas at Dallas, 800 W. Campbell Rd, Richardson, TX, USA. Email: {\tt\small (vrengana, nxh150030, tyler.summers, jruths)@utdallas.edu}}%
}

\begin{document}
\maketitle
\begin{abstract}
Designing resilient control strategies for mitigating stealthy attacks is a crucial task in emerging cyber-physical systems. In the design of anomaly detectors, it is common to assume Gaussian noise models to maintain tractability; however, this assumption can lead the actual false alarm rate to be significantly higher than expected. We propose a distributionally robust anomaly detector for noise distributions in moment-based ambiguity sets. We design a detection threshold that guarantees that the actual false alarm rate is upper bounded by the desired one by using generalized Chebyshev inequalities. Furthermore, we highlight an important trade-off between the worst-case false alarm rate and the potential impact of a stealthy attacker by efficiently computing an outer ellipsoidal bound for the attack-reachable states corresponding to the distributionally robust detector threshold. We illustrate this trade-off with a numerical example and compare the proposed approach with a traditional chi-squared detector.
\end{abstract}

\section{INTRODUCTION} \label{sec_intro}
Many emerging complex dynamical networks, from critical infrastructures to industrial cyber-physical systems (CPS) and various biological networks, are increasingly able to be instrumented with new sensing and actuation capabilities. These networks comprise growing webs of interconnected feedback loops and must operate efficiently and resiliently in dynamic and uncertain environments. As these systems become large, devising automated methods for detecting anomalies (such as component failures or malicious attacks) are critical for smooth and efficient operation. Such critically important cyber-physical networks have become an attractive target to attackers. These systems are large and complex enough -- and often not monitored well enough -- for attackers to manipulate the system without being detected and cause significant damage \cite{mo2012integrity, nateghi_paper, Cardenas_Sipolini,Pasqualetti_1, Cardenas,Mo_3,Carlos_Justin3,navid,Mo_2,karl_johannson_tcst}. 

To simplify analysis and design, often such complex cyber-networks are modeled as a discrete-time linear time invariant system. In such systems, noises are often modelled as Gaussian in the name of tractability. However, this can lead to a significant miscalculation of probabilities and risk if the underlying processes behave differently, for example due to various nonlinearities or malicious attacks. In the context of attacks, it is possible for an attacker to modify the sensor outputs and effectively generate aggressive and strategic noise profiles to sabotage the operation of the system. In stochastic optimization, these limitations are being recognized and addressed in the emerging area of distributionally robust optimization (DRO) \cite{dr_wiesemann}.


Here, we propose to use distributionally robust optimization (DRO) methods to improve modeling and reduce false alarm rates in cyber-physical networks. DRO enables modelers to explicitly incorporate inherent ambiguity in probability distributions into optimization problems. This more realistic account of uncertainty mitigates the so-called optimizer’s curse, where overly strong assumptions about uncertainty distributions can lead to poor out-of-sample performance. Moreover, several important DRO problems can be tractably solved. The central idea is to represent uncertainty through an ambiguity set as a family of possibly infinitely many probability distributions consistent with the available raw data or prior structural information, and to model the decision-making process as a game against ``nature''. In this game, the modeler first selects a decision with the goal to optimize his objective or maximize the probability of constraint satisfaction, in response to which nature selects a distribution from within the ambiguity set with the goal to inflict maximum harm to the modeler. This game mimics the adversarial nature of attacks and makes DRO an obvious choice to analyze worst case behavior not just from noise/nature, but from corruption by an attacker.

A model-based approach to attack detection uses a detector that raises alarms when there is a large enough discrepancy between the actual measurements and an estimate of the system, a statistic termed the residual. The detector's sensitivity can be increased by decreasing the threshold of detection, but there is an inherent tradeoff between sensitivity and the rate at which false alarms are generated. Keeping false alarms to a manageable level requires decreasing sensitivity and the tuning of the detector threshold is typically informed by the distribution of the residual. When this distribution is only known to an ambiguity set, traditional tools no longer suffice to select the threshold and so we turn to a distributionally robust approach. In the context of attacks, the tuning of the detector has a direct implication on the effect an attacker can have while still remaining stealthy. 


\textit{Contributions:} Our main contributions are: 1) design of a detection threshold that guarantees that the actual false alarm rate is upper bounded by the desired one by using generalized Chebyshev inequalities (Proposition \ref{prop_dr_alfa}); 2) formulation of ellipsoidal outer bounds on the reachable sets of the system corresponding to attacks despite the exact distribution of the noise being unknown using distributional robustness (Proposition \ref{prop:reachableset}); 3) demonstrating via a numerical example an important trade-off between the worst-case false alarm rate and the potential impact of a stealthy attacker. Specifically, we use generalized Chebyshev inequalities \cite{boyd_sdp,chen_chebyshev}, to find the detector threshold values so that the anomaly detector we design achieves a desired worst-case false alarm rate. Finally, using the optimum detector threshold values, we derive the outer bounding ellipsoid that contains the reachable set induced by a stealthy attack by solving a convex optimization problem. While anomaly and attack detection is widely stuied in CPS literature, our distributionally robust approach marks the novel contribution of this paper.  

The rest of the paper is organized as follows. Section \ref{sec_problem_formulation} formulates the problem statement and explains the distributionally robust approach and further using the generalized Chebyshev bounds to design anomaly detectors. Section \ref{sec_reach_sets} describes the convex optimization problem formulation to find an ellipsoidal bound on the reachable sets obtained using distributionally robust tuned detector. Section \ref{sec_numerical_simulation} discusses about the numerical results using an empirical system and highlights the trade-offs observed between the attacker's capability and being distributionally robust against any noise distribution. Finally, Section \ref{sec_conclusion} concludes and summarizes future research directions.

\section{Problem Formulation and Distributionally Robust Anomaly Detector Tuning} \label{sec_problem_formulation}
We model an uncertain cyber-physical system using a stochastic discrete-time linear time invariant
(LTI) system
\begin{equation}
\begin{aligned}
    x_{t+1} &= A x_t + B u_t + w_t, \quad t \in \mathbb{N} \\
    y_t &= C x_t + v_t,
\end{aligned}
\end{equation}
where $x_t \in \mathbb{R}^n$ is the system state at time $t$, $u_t \in \mathbb{R}^m$ is the input at time $t$, $A$ is the dynamics matrix, $B$ is the input matrix. The process noise $w_t$ is modeled using a zero-mean random vector independent and identically distributed across time with covariance matrix $\Sigma_{w}$. The output $y_t \in \mathbb{R}^p$ aggregates a linear combination, given by the observation matrix $C \in \mathbb{R}^{p \times n}$, of the states and the sensor noise, $v_k$ modeled using a zero-mean random vector independent and identically distributed across time with covariance matrix $\Sigma_{v}$. The distributions $P_w$ of $w_t$ and $P_v$ of $v_t$ are unknown (and not necessarily Gaussian) and will be assumed to belong to the ambiguity sets $\mathcal{P}^w, \mathcal{P}^v$ of distributions respectively. With the second moments of the process noise and sensor noise denoted by $\Sigma_{w} = \mathbf{E} [w_t w^{\top}_t]$ and $\Sigma_{v} = \mathbf{E} [v_t v^{\top}_t]$ being known, we can then define the moment based ambiguity sets as follows,
\begin{equation} \label{eqn_ambig_set_v}
    \mathcal{P}^{v} = \{ P_{v} \, | \, \, \mathbf{E}v_t = 0, \, \, \mathbf{E} [v_t v^{\top}_t] = \Sigma_{v} \}, \\ 
\end{equation}
\begin{equation} \label{eqn_ambig_set_w}
    \mathcal{P}^{w} = \{ P_{w} \, | \, \, \mathbf{E}w_t = 0, \, \, \mathbf{E} [w_t w^{\top}_t] = \Sigma_{w} \}. \\ 
\end{equation}
We assume that the pair $(A,C)$ is detectable and $(A,B)$ is stabilizable. In this work, we consider the scenario that the actual measurement $y_t$ can be corrupted by an additive attack, $\delta_t \in \mathbb{R}^p$. Due to this additive attack, the output of the system fed to the controller becomes
\begin{equation}
    \Bar{y}_t = y_t + \delta_t = C x_t + v_t + \delta_t.
\end{equation}
To leverage a fault-detection approach, we require an estimator of some type to produce a prediction of the system behavior. In this work, we use the steady state Kalman filter
\begin{equation}
    \hat{x}_{t+1} = A \hat{x}_t + B u_t + L (\Bar{y}_t - C \hat{x}_t),
\end{equation}
where $\hat{x}_t \in \mathbb{R}^n$ is the estimated state. The observer gain $L$ is designed to minimize the steady state covariance matrix $P$ in the absence of attacks, where the estimation error is
\begin{equation}
    \begin{aligned}
    e_t &= x_t - \hat{x}_t, \quad \text{and} \\
    P &:= \lim_{t \rightarrow \infty} P_t := \mathbf{E}[e_t e^{\top}_t]. 
    \end{aligned}
\end{equation}
Recall that $P$ is the solution of an algebraic Ricatti equation.
Since $(A,C)$ is assumed to be detectable, the existence of such a steady state covariance matrix $P$ is guaranteed. Now, we define a residual sequence, $r_t$ as the difference between what we actually receive $\Bar{y}_t$ and expect to receive $C \hat{x}_t$ as, 
\begin{equation}
    \begin{aligned}
        r_{t} &= \Bar{y}_t - C \hat{x}_t = C e_t + v_t + \delta_t,
    \end{aligned}
\end{equation}
and the estimation error evolves according to 
\begin{equation} \label{eqn_error_dynamics}
    \begin{aligned}
        e_{t+1} &= (A - LC) e_t - L v_t - L \delta_t.
    \end{aligned}
\end{equation}
When there is no attack, that is, $\delta_t = 0,$ the residual sequence $r_t$ falls according to a zero mean distribution with covariance 
\begin{equation} \label{eqn_residual}
    \Sigma_{r} = \mathbf{E}[r_t r^{\top}_t] = C P C^{\top} + \Sigma_{v}.
\end{equation}
Note that $r_t$ is a zero-mean random vector independent and identically distributed across time with covariance matrix $\Sigma_{r}$ and the distribution $P_r$ of $r_t$ is unknown (and not necessarily Gaussian) but belongs to an ambiguity set $\mathcal{P}^r$ whose second moment, $\Sigma_r$, can be calculated from \eqref{eqn_residual}.
\begin{equation} \label{eqn_ambig_set_r}
    \mathcal{P}^{r} = \{ P_{r} \, | \, \, \mathbf{E}r_t = 0, \, \, \mathbf{E} [r_t r^{\top}_t] = \Sigma_{r} \}.
\end{equation} 
\subsection{Distributionally Robust Optimization Approach}
Distributionally robust optimization approaches can be categorized based on the form of the ambiguity set. There are several different parameterizations, including those based on moments, support, directional derivatives \cite{dr_goh}, and Wasserstein balls \cite{dr_peyman}. For example, a moment-based ambiguity set includes all distributions with a fixed moments up to some order (e.g., fixed first and second moments), and Wasserstein-based ambiguity sets include a ball of distributions within a given Wasserstein distance from some base distribution (such as an empirical distribution on a training data set). We will focus here on the moment-based ambiguity set as explained in \cite{dr_wiesemann}, though other parameterizations are interesting and relevant for bounding reachable sets using ellipsoids will be pursued in future work. Similar in structure to a chi-squared detector though tuned using a distributionally robust approach, we define a quadratic distance measure $z_t$ to be sensitive to changes in the variance of the distribution as well as the expected value, 
\begin{equation}
    z_t = r^{\top}_t \Sigma^{-1}_{r} r_t.
\end{equation}
Notice that $z_t$ is also a random variable expressed as the sum of the squares of $r_t$. In the case that the residual is Gaussian, $z_t$ would be a chi-squared random variable (hence the name of the chi-squared detector).

\subsection{Anomaly Detector Thresholds and False Alarm Rates}
For a given threshold $\alpha \in \mathbb{R}_{>0}$ and the distance measure $z_t = r^{\top}_t \Sigma^{-1}_{r} r_t$, 
\begin{equation}
    \begin{aligned}
        \begin{cases} z_t \leq \alpha, &\text{no alarm} \\
        z_t > \alpha, &\text{alarm: } t^{*} = t.
        \end{cases}
    \end{aligned}
\end{equation}
alarm time(s) $t^{*}$ are produced. Due to the infinite support of the sensor noise $v_t$, the distance measure $z_t$ will also have infinite support. Thus even in the absence of attacks, the detector is expected to generate false alarms because some values drawn from the distance measure distribution will exceed the threshold $\alpha$. Usually the detectors are designed for a desired false alarm rate, $\mathcal{A}$, through an appropriate choice of threshold $\alpha$. This, however, requires knowing the distribution governing the detector random variable, $z_t$. If the distribution of the quadratic distance measure $z_t$ is known, then it is possible to extract the optimum threshold values $\alpha$ from the knowledge of the distribution. Suppose for instance, if the traditional chi-squared detector is used as in \cite{navid}, with threshold $\alpha \in \mathbb{R}_{>0}$, $r_t \sim \mathcal{N}(0,\Sigma_{r})$. Then corresponding to the desired false alarm rate $\mathcal{A} = \mathcal{A}^{*}$, we can obtain the optimum threshold as 
\begin{equation} \label{eqn_chi_squared_threshold}
\alpha = \alpha^{*} := 2 \mathbf{P}^{-1}\left(1-\mathcal{A}^{*}, \frac{p}{2}\right),
\end{equation}
where $\mathbf{P}^{-1}(\cdot, \cdot)$ denotes the inverse regularized lower incomplete gamma function.

\subsection{Tuning the Threshold via Generalized Chebyshev Bounds}
When the complete distribution is not available, tuning methods like the one above may design thresholds that generate actual false alarm rates significantly higher than what is desired. With the distributionally robust approach, we aim to achieve a desired false alarm rate $\mathcal{A}$, and the challenge is to select $\alpha$ such that
\begin{equation} \label{eqn_r_ineq}
    \begin{aligned}
    \sup_{P_{r} \in \mathcal{P}^{r}} P_{r}(r^{\top}_{t} \Sigma^{-1}_{r} r_{t} \leq \alpha) &= 1 - \mathcal{A}.
    \end{aligned}
\end{equation}
\begin{proposition} \label{prop_dr_alfa}
Given a desired false alarm rate $\mathcal{A}$ and the true distribution of the system residual $r_t$ belonging to an ambiguity set $\mathcal{P}^{r}$ defined as in \eqref{eqn_ambig_set_r}, the distributionally robust detector threshold $\alpha$ satisfying \eqref{eqn_r_ineq} is given by 
\begin{align} \label{eqn_DR_alpha}
    \alpha = \frac{p}{\mathcal{A}},
\end{align}
where $p$ denotes the number of outputs.
\end{proposition}
\begin{proof}
A sharp lower bound on the probability of a set defined by quadratic inequalities, given the first two moments of the distribution can be efficiently computed using a semidefinite program described in \cite{boyd_sdp}, which generalizes Chebyshev's inequality to vector random variables. However, for the particular form of the residual set in \eqref{eqn_r_ineq} there is an analytical solution \cite{chen_chebyshev}. Using this result, we can obtain the Chebyshev bound $\alpha$ satisfying \eqref{eqn_r_ineq} as follows,
\begin{equation} \label{eqn_Chebyshev_inequality}
    \begin{aligned}
    \sup_{P_{r} \in \mathcal{P}^{r}} P_{r}(r^{\top}_{t} \Sigma^{-1}_{r} r_{t} \geq \alpha) &\leq \frac{p}{\alpha}, 
    \end{aligned}
\end{equation} 
where $p$ is the number of outputs. Comparing \eqref{eqn_Chebyshev_inequality} with the \eqref{eqn_r_ineq}, we can see that
\begin{equation*} 
    \begin{aligned}
    \frac{p}{\alpha} = \mathcal{A} \implies \alpha = \frac{p}{\mathcal{A}}. 
    \end{aligned}
\end{equation*}

\end{proof}


\section{Attack-Reachable Set Bounds}\label{sec_reach_sets}
The threshold of the anomaly detectors limit what the attacker is able to accomplish if he/she seeks to remain undetected. These attack models require strong attacker knowledge and access, namely we assume that the attacker has perfect knowledge of the system dynamics, the Kalman filter, control inputs, and measurements. In addition, the attacker has read and write access to all the sensors at each time step. In this section, we describe a stealthy attack by an attacker and define reachable set to quantify the system impact due to the attack and process noise.  
\subsection{Zero Alarm Attacks}
Zero-alarm attacks generate attack sequences that maintain the distance measure at or below the threshold $z_t \leq \alpha$, so that no alarms are raised during the attack. To satisfy this condition we define the attack as 
\begin{equation} \label{eqn_attack_seq}
    \delta_t = - Ce_t - v_t + \Sigma^{\frac{1}{2}}_{r} \bar{\delta}_t 
\end{equation}
where $\Sigma^{\frac{1}{2}}_{r}$ is the symmetric square root of $\Sigma_{r}$ and $\bar{\delta}_t \in \mathbb{R}^p$ is any vector that the attacker has the freedom to design such that $\bar{\delta}^{\top}_t \bar{\delta}_t \leq \alpha$. With such an attack sequence, the distance measure becomes
\begin{equation} \label{eqn_zero_alarm_attack}
    z_t = r^{\top}_t \Sigma^{-1}_{r} r_t = \bar{\delta}^{\top}_t \bar{\delta}_t \leq \alpha.
\end{equation} 
\subsection{Computing Ellipsoidal Bound for the Reachable Set}
In order to compare the effects of an attack, we require a metric to quantify the impact of it. A popular choice to quantify system impact due to a disturbance is the set of states reachable by the action of the disturbance. We intend to find a best bound for $\alpha$ that will result in the required user prescribed false alarm rate $\mathcal{A}$, keeping in mind the noise defining $z_t$ may fall under any distribution from its corresponding moment-based ambiguity set. The generalized Chebyshev bound $\alpha$ addressing the above problem is obtained by \eqref{eqn_DR_alpha}. When the distribution of the residual is not known exactly, the lack of information inherently leads to more conservatism, making the value of $\alpha$ larger, and therefore, the magnitude of zero-alarm attacks larger. 

When there is an attack as in \eqref{eqn_attack_seq}, the evolution of the system dynamics can be written in a new and reduced form where the measurement noise $v_t$ gets cancelled and thus resulting in the state and estimation error dynamics of the system as a function of $w_t$ and $\bar{\delta}_t$. 
Defining the joint state as $\xi_t=\begin{bmatrix}x_t &  e_t\end{bmatrix}^{\top}$ with input $\zeta_t=\begin{bmatrix}w_t & \bar{\delta}_t\end{bmatrix}^{\top}$, we can write the joint evolution as
\begin{equation}\label{eqn_stacked_states}
\xi_{t+1}=\hat{A}\xi_{t}+\hat{B} \zeta_t, 
\end{equation}
where $\hat{A}=\begin{bmatrix}A+BK & -BK \\ 0 & A\end{bmatrix}$ and $\hat{B}=\begin{bmatrix}I & 0\\I & -L\Sigma_r^{1/2}\end{bmatrix}$. \\
Since some distributions in the ambiguity set of $\mathcal{P}^w$ may have unbounded support, it is necessary for us to truncate them at some confidence level (since unbounded noise would, theoretically, lead to unbounded reachable sets, albiet for infinitesimal probabilities). We follow the distributionally robust approach using \eqref{eqn_DR_alpha} and \eqref{eqn_Chebyshev_inequality} to obtain noise threshold $\bar{w}$ satisfying,  
\begin{equation}\label{eqn_DRnoise_threshold}
    \sup_{P_{w} \in \mathcal{P}^{w}} P_{w}(w^{\top}_{t} \Sigma^{-1}_{w} w_{t} \leq \bar{w}) = 1 - \mathcal{A}.
\end{equation}
The reachable set of interest, driven by the ellipsoidally bounded inputs $w_t$ and $\bar{\delta}_t$, is
\begin{equation}
\mathcal{R}_x = \left\{  x_t=[I_n,\, 0_{n\times n}]\xi_t \ \left|\
    \begin{aligned}
        &\xi_{t+1}=\hat{A}\xi_{t}+ \hat{B} \zeta_t, \\
        &\xi_1=\mathbf{0},\ \delta_t^\top\delta_t \leq \alpha,\\
        &w_t^\top\Sigma_w^{-1}w_t \leq \bar{w},\ \forall t\in\mathbb{N}
    \end{aligned}
    \right. 
    \right\}.
\end{equation}
We use Linear Matrix Inequalities, for some positive definite matrix $Q_x$, to derive outer ellipsoidal bounds of the form 
\begin{equation}
    \mathcal{R}_{x} \subseteq \mathcal{E}(Q_{x}) = \{ x_t \, | \, x^{\top}_{t} \mathcal{P} x_t \leq 1\}, 
\end{equation}
where the ellipsoid $\mathcal{E}$ has minimum volume and shape matrix $Q_x$. We define $\mathcal{P}_\xi$ as the inverse of the shape matrix of the ellipsoidal bound for the $\xi$ reachable set, 
\begin{equation} \label{eqn_psi_reach_set_matrix}
    \mathcal{P}^{-1}_\xi = Q_{\xi} = \begin{bmatrix}Q_x & Q_{xe} \\ Q_{xe}^{\top} & Q_e\end{bmatrix}.
\end{equation}
The following proposition will introduce the optimization problem to find the minimum volume ellipsoidal bound for the reachable set. 
\begin{proposition} \label{prop:reachableset} Given the system matrices $A,B,C$, gain matrices $K, L$, a positive semi-definite matrix $\mathcal{F}$, attack input threshold $\alpha$, system noise threshold $\bar{w}$, and constant $a \in [0,1)$ the following convex optimization generates the smallest reachable set ellipsoidal bound $\mathcal{E}(Q_x)$,
\begin{equation} \label{eqn_sdp_alpha_bar}
    \begin{aligned}
        &\underset{a_1, a_2, Q_x , Q_{xe}, Q_e}{\text{minimize}} & & \textbf{tr}(Q_x)\\
        &\text{\quad subject to } & &  a_1 + a_2 \geq a, \quad  a_1,a_2 \in [0,1)\\ 
        & & & Q_\xi \geq 0 , \  \mathcal{F} \geq 0.
    \end{aligned}
\end{equation}
\end{proposition} 
\noindent \begin{proof}
In order to prove the proposition, we leverage the results from \cite{Murguia2018SecurityMO} and Lemma 1 in \cite{Carlos_Justin3} (restated below). 
\begin{lemma} \label{fatlemma}
Let $V_t$ be a positive definite function, $V_1 = 0$, $\zeta_{it}^{\top} W_i \zeta_{it} \leq 1$, $i = 1 \dots N$, where $N$ is the number of available inputs and $W_i$ is the inverse of shape matrix for the ellipsoidal bound over input $\zeta_{it}$ and is positive definite. Then, it can be shown that $V_t \leq \frac{N-a}{1-a}$, if there exists a constant $a \in (0,1)$ and $a_i \in (0,1), \forall i = 1,\dots,N$ such that $\sum_{i=1}^{N} a_i \geq a$ and
\begin{equation} \label{eqn_pos_def_cdtn} 
V_{t+1} - aV_t -\sum_{i=1}^{N} (1-a_i) \zeta_{it}^{\top} W_i\zeta_{it} \leq 0.
\end{equation}
\end{lemma}
\noindent The proof of the lemma is available in \cite{Murguia2018SecurityMO}. To derive the reachable set bound of $\xi_t$, let us define the positive definite function required in \eqref{eqn_pos_def_cdtn} as follows,
\begin{equation} \label{eqn_pos_def_fn}
V_t=\xi_t^{\top} \Tilde{\mathcal{P}}_{\xi} \xi_t \leq \frac{2-a}{1-a},\ \  \Tilde{\mathcal{P}}_{\xi}>0,\ \  \mathcal{P}_\xi=\frac{1-a}{2-a} \Tilde{\mathcal{P}}_{\xi}.    
\end{equation}
Substituting \eqref{eqn_pos_def_fn} in \eqref{eqn_pos_def_cdtn}, we get
\begin{equation}\label{eq:errorlemma1}
V_{t+1} - aV_t - \frac{1-a_2}{\alpha} \bar{\delta}_t^{\top} \bar{\delta}_t -\frac{1-a_1}{\bar{w}}w_t^{\top}\Sigma_w^{-1}w_t \leq 0,
\end{equation}
and further solving it using Schur complement results in the following linear matrix inequality
\begin{equation}\label{eqn_boundLMI1}
    \mathcal{H}=\begin{bmatrix}a\mathcal{P}_\xi &  \hat{A}^{\top}\mathcal{P}_\xi &0\\
    \mathcal{P}_\xi \hat{A} & \mathcal{P}_\xi & \mathcal{P}_\xi  \hat{B}\\
    0 &  \hat{B}^{\top}\mathcal{P}_\xi & \hat{W}
    \end{bmatrix} \geq 0,
\end{equation}
where, $\hat{W} = \frac{1-a}{2-a}W_{a}$ and  $W_{a}=\begin{bmatrix}\frac{1-a_1}{\bar{w}}\Sigma_w^{-1} & 0 \\ 0 & \frac{(1-a_2)}{\alpha}I_{p}\end{bmatrix}$. To replace $\mathcal{P}_\xi$ with $Q_\xi$, we use diagonal transformation matrix $\Theta$ such that, $\mathcal{F} = \Theta^{\top} \mathcal{H} \Theta$ where, $\Theta =$ diag$(Q_\xi, Q_\xi, I)$. Now, \eqref{eqn_boundLMI1} equivalently gets transformed into  
\begin{equation} \label{eqn_boundLMI2}
\mathcal{F} = \begin{bmatrix}a Q_\xi &  Q_\xi {A}^{\top} & 0 \\ A Q_\xi & Q_\xi & B \\ 0 &  {B}^{\top} & \hat{W} \end{bmatrix} \geq 0.
\end{equation}
Thus $\mathcal{F} \geq 0$ with $Q_{\xi} \geq 0$ will ensure that the convex optimization problem given by \eqref{eqn_sdp_alpha_bar} will result in an ellipsoid bounding the reachable set. There are many choices for the objective function to tighten the outer ellipsoid bound, but the trace criteria tends to find compact ellipsoids without a large principal axis. 
\end{proof}

\section{Numerical Simulation} \label{sec_numerical_simulation}
In this section, we demonstrate the performance of the distributionally robust fault detector when there is no attack on the system. Using the same detector we present our analysis of the effects of stealthy attacks on the system by studying the ellipsoids that bound the reachable states. We consider the following system for this study with the detector tuned to a false alarm rate $\mathcal{A} = 0.05$  ($5$\%):
\begin{equation*}
    \begin{aligned}
        A &= \begin{bmatrix} 0.84 & 0.23 \\ -0.47 & 0.12 \end{bmatrix}, B = \begin{bmatrix} 0.07 & -0.32 \\ 0.23 & 0.58\end{bmatrix}, C = \begin{bmatrix} 1 & 0 \\ 2 & 1 \end{bmatrix} \\
        K &= \begin{bmatrix} 1.404 & -1.402 \\ 1.842 & 1.008 \end{bmatrix}, L = \begin{bmatrix} 0.0276 & 0.0448 \\ -0.01998 & -0.0290 \end{bmatrix}, \\
        \Sigma_v &= 2 I_p, \Sigma_w = \begin{bmatrix} 0.045 & -0.011 \\ -0.011 & 0.02 \end{bmatrix}.
    \end{aligned}
\end{equation*}  
\subsection{Advantages of Distributional Robustness}
The purpose of this first simulation is to demonstrate the effectiveness of distributionally robust approach while comparing it with the traditional chi-squared detector approach which assumes a normal distribution for the noises. We assume that there is no attack on the system so that the estimation error evolves according to \eqref{eqn_error_dynamics} but with $\delta_t = 0$. We run an extensive Monte-Carlo simulation to derive an empirical probability of the distance measure, $z_t$, lying above a threshold. To demonstrate the effectiveness of our proposed approach, we investigate two detectors: 1) a detector tuned assuming the noises are Gaussian distributed with threshold $\alpha_{\chi^{2}}$ using \eqref{eqn_chi_squared_threshold}; and 2) a detector tuned making no assumption about the distribution of the noises with threshold $\alpha_{DR}$ using \eqref{eqn_DR_alpha}. For demonstration, we test these detectors with two scenarios: 1) a scenario where both noises are Gaussian and 2) a scenario where both noises are distributed according to Student's $t$ distribution (having heavy tails) with the degree of freedom $\nu = 5$. In both cases the noises belong to the ambiguity sets $\mathcal{P}^v$ and $\mathcal{P}^w$ characterized by zero-mean and respective covariances $\Sigma_v$ and $\Sigma_w$ listed above.  Finally, we estimate the false alarm rate of each detector under each noise scenario, by evaluating the probability of $z_t$ falling above  each detector threshold.
\begin{figure}[t]
    \centering
    \includegraphics[scale=0.18]{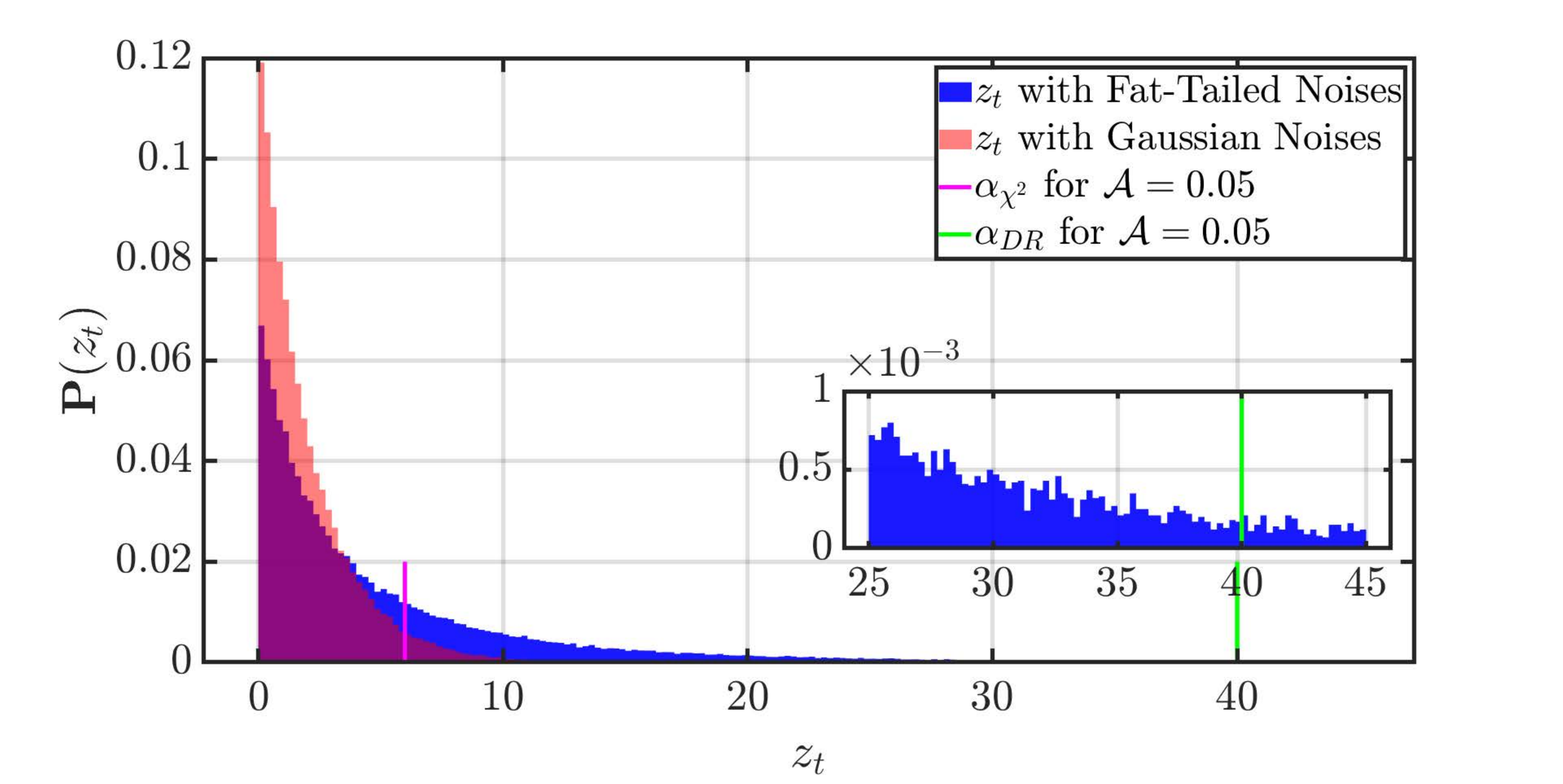}
    \caption{The histogram plot shows the probability of the quadratic distance measure random variable $z_t$ in two noise scenarios where in the first case the system is driven by Gaussian noises and in the second case the system is driven by fat-tailed noises. Two detectors are tuned using \eqref{eqn_chi_squared_threshold} and \eqref{eqn_DR_alpha} respectively. While the chi-squared detector generated $5\%$ false alarms for the Gaussian noise scenario, it generated $29.81 \% $ when the noises were fat-tailed. The distributionally robust detector generated $0 \%$ false alarms for the Gaussian noise scenario and $0.83\%$ false alarms when the noises were fat-tailed. The inner plot shows the zoomed in part for $z_t \in [25,45]$.}
    \label{fig_normalizedresidual}
\end{figure}
The result of a Monte-Carlo simulation with 100,000 trials (for each noise scenario) is shown as an histogram plot of the distance measure of $z_t$ in Fig. \ref{fig_normalizedresidual}. In the first noise scenario, when the noises are Gaussian, the distributionally robust detector provides a conservative threshold bound with $0\%$ false alarm rate. The chi-squared detector resulted in the user prescribed false alarm rate of $5\%$. 

In the second noise scenario, we see an increased false alarm rate with the traditional chi-squared detector. Due to the wrong assumption for system noise and sensor noise, the chi-squared detector generated $29.81\%$ false alarms resulting in significantly miscalculated risk. The distributionally robust detector generated $0.83\%$ false alarms. Since our goal in tuning the detector is to create a monitor that generates false alarms no more than $\mathcal{A} = 5\%$ of the time, we see that the DR approach achieves this aim, while the traditional chi-squared approach does not.

\textbf{Remark 1:} While the low false alarm rate of the distributionally robust detector looks appealing, it allows a malicious attacker to execute stealthy attacks with larger impact which will be demonstrated in the following simulations. 

\textbf{Remark 2:} The performance of the distributionally robust detector can be improved significantly if higher order moments or other structural information about the distributions (e.g., symmetry or unimodality) of the residual sequence are also utilized. This will result in sharper probability estimates as mentioned in \cite{bertsimas_popescu} and hence in tighter threshold values.

\subsection{Cost of Being Distributionally Robust} 
\begin{figure}[t]
    \centering
    \includegraphics[scale = 0.18]{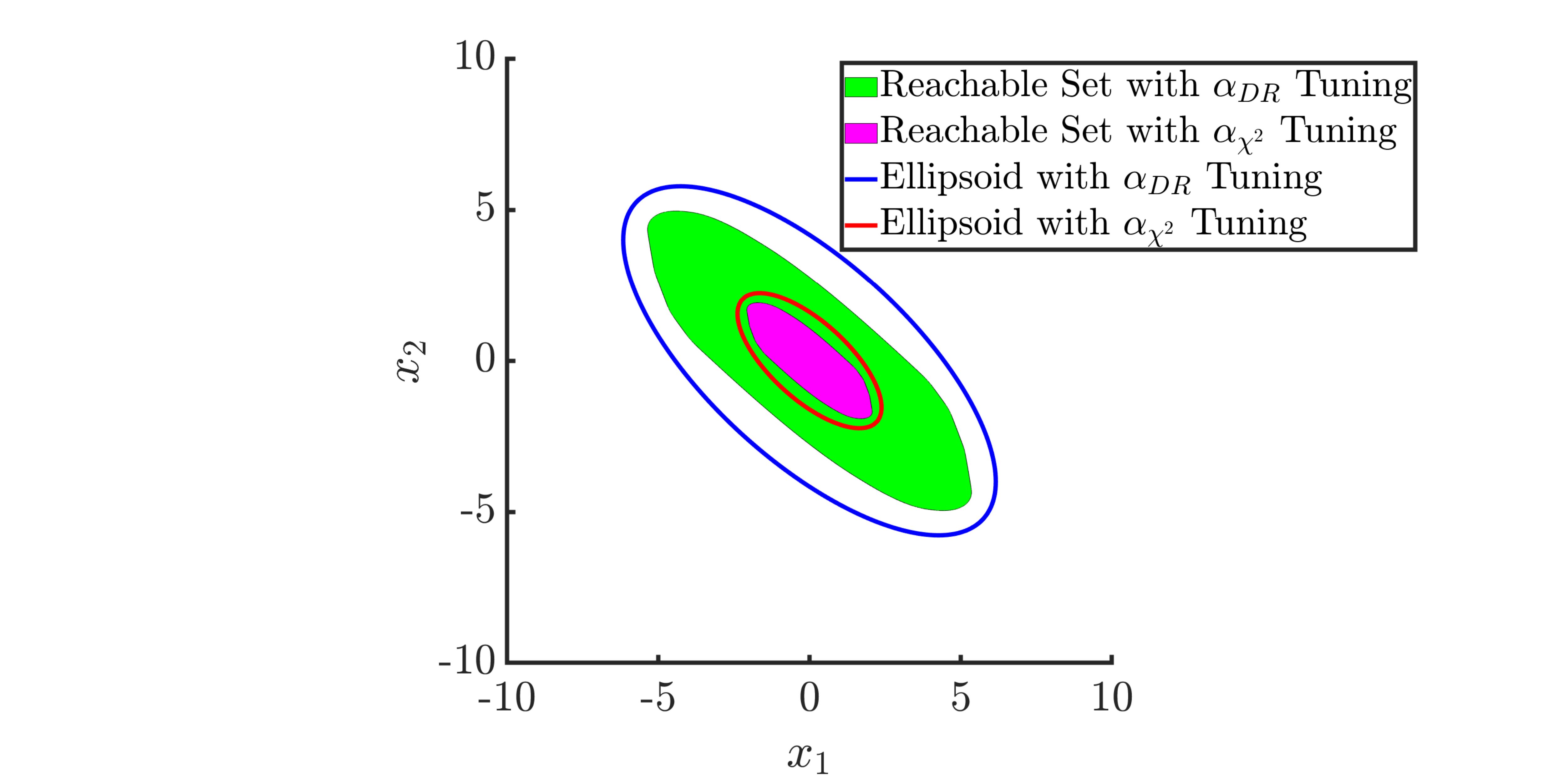}
    \caption{Two empirical reachable states when the system is driven by zero-alarm attacks and system noise: one larger (green) corresponding to a distributionally robust tuned detector and another smaller (magenta) corresponding to a chi-squared tuned detector. Their respective outer bounding ellipsoids, which can be efficiently calculated with Proposition \ref{prop:reachableset}, in blue and red. 
    }
    \label{fig_all_in_all} 
\end{figure}
\begin{figure}[t]
    \centering
    \includegraphics[scale=0.18]{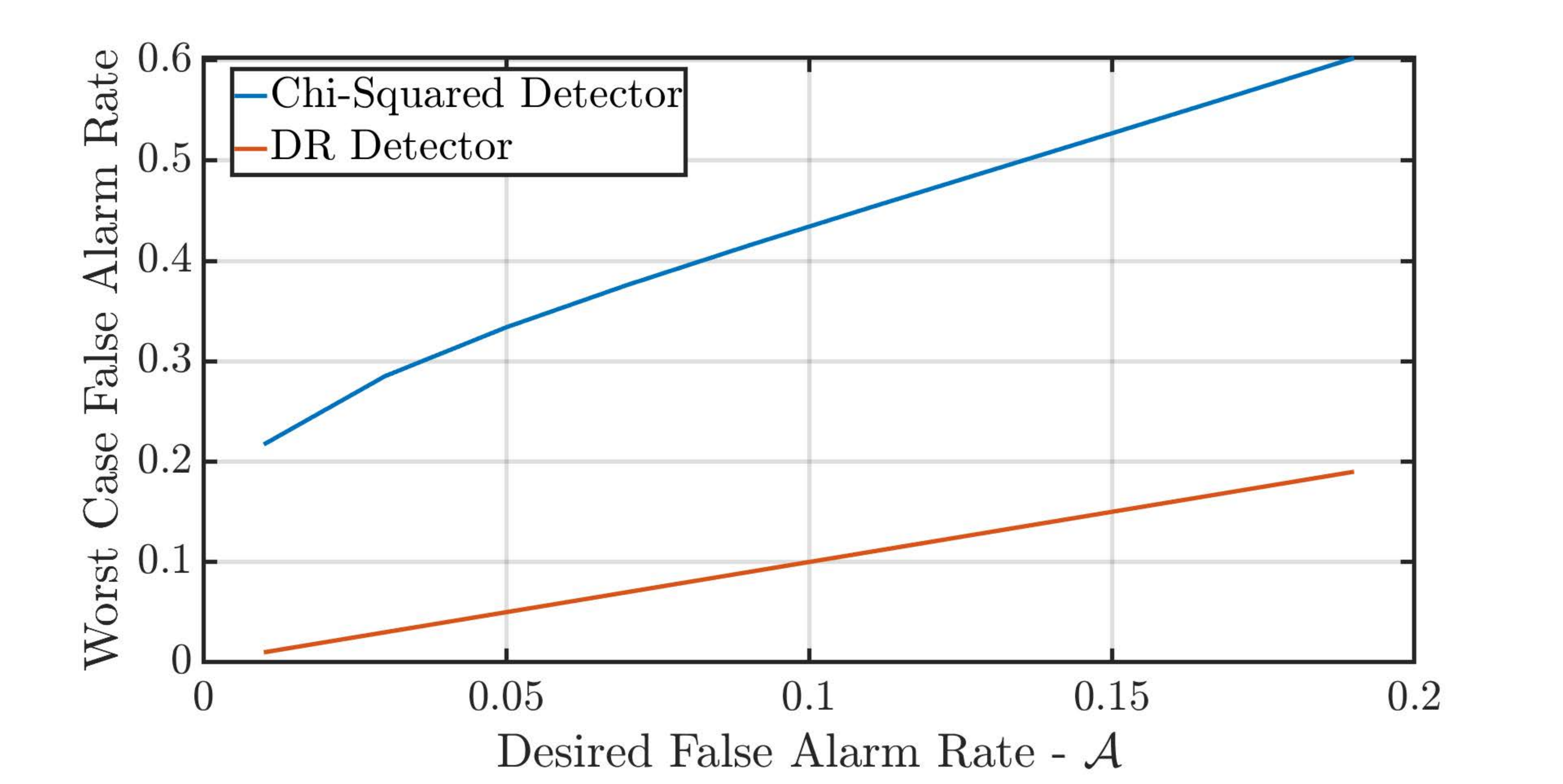}
    \caption{The variation of worst case false alarm rate as a function of the user prescribed desired false alarm rate is shown here. The markers in both the lines along the same vertical axis correspondingly represent the same threshold. It is evident that a superlinear behavior (blue curve) is observed depicting that worst case false alarm rate is bigger than the desired false alarm rate for a given threshold.}
    \label{fig_worst_A}
\end{figure}
\begin{figure}[t]
    \centering
    \includegraphics[scale=0.18]{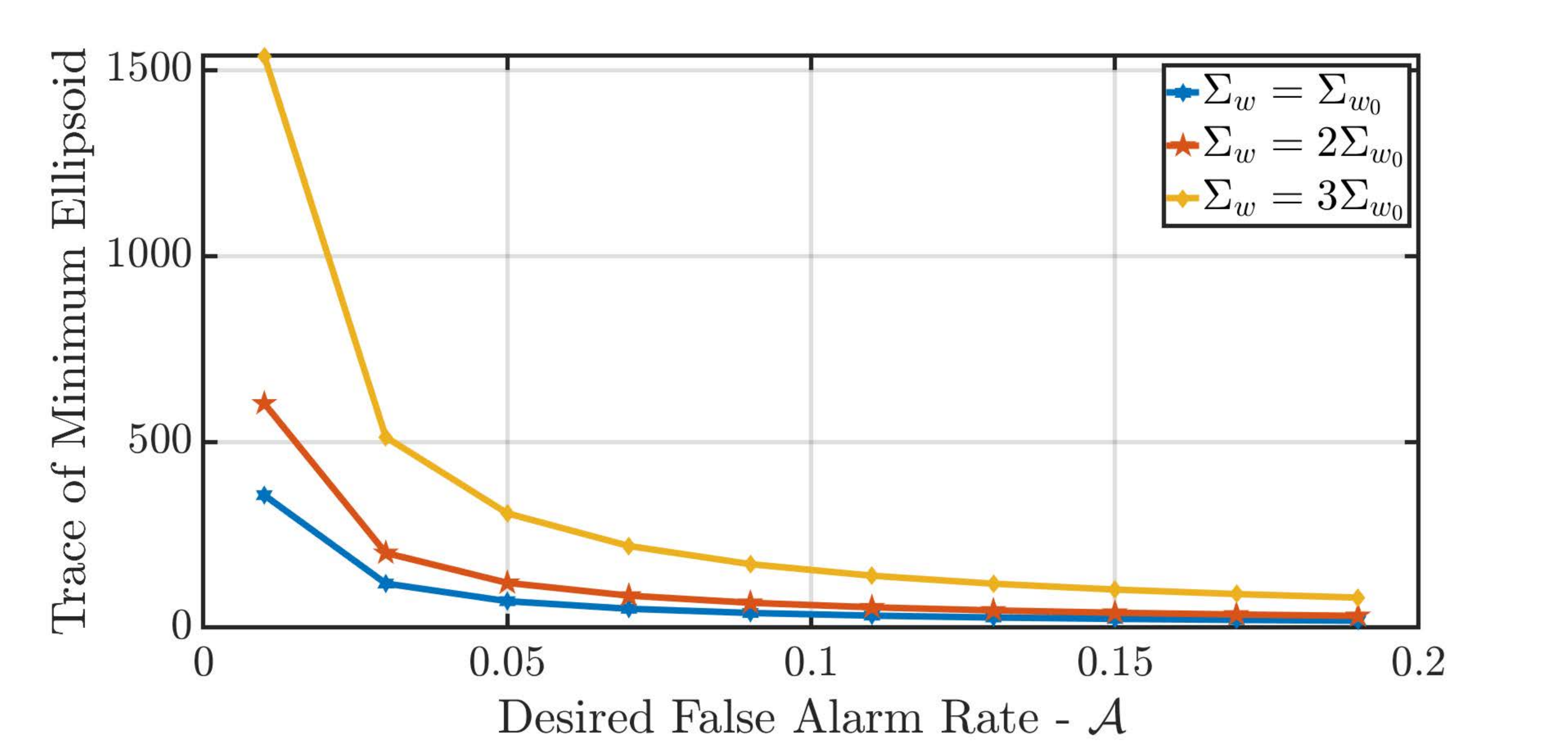}
    \caption{The trade-off shown here indicates that as the distributionally robust detector is tuned for larger false alarm rates, the trace of the bounding ellipsoid obtained using the $\alpha_{DR}$ tuning decreases. This trend is inversely proportional, indicating tuning the distributionally robust detector for a very small false alarm rate is very expensive in terms of the size of the reachable set.  
    }
    \label{fig_vol_vs_alarm_rate}
\end{figure}
We will now demonstrate the trade-off between being distributionally robust against any noise distribution and the increased attacker capabilities caused by the conservative robust tuning. We quantify this trade-off through reachable set analysis.  We continue to consider the two detectors: one tuned assuming the noises are Gaussian and one tuned without any distributional assumptions through the robust tuning presented in this paper. We study the reachable set that an attacker is able to accomplish with each detector, while remaining stealthy, i.e., the attack input satisfies the zero-alarm stealthiness criteria in \eqref{eqn_zero_alarm_attack}, and, therefore, do not raise alarms. We quantify the reachable set through the outer ellipsoidal bounds found by the optimization in Proposition \ref{prop:reachableset}.

Fig. \ref{fig_all_in_all} shows the outcome of the two zero-alarm attacks, one made stealthy to the chi-squared detector and one made stealthy to the distributionally robust detector. The chi-squared ellipsoid is obtained as a function of the detector threshold computed from \eqref{eqn_chi_squared_threshold}. Similarly, the distributionally robust ellipsoid is obtained as a function of the distributionally robust detector threshold computed from \eqref{eqn_DR_alpha}. As $\alpha_{DR}\geq \alpha_{\chi^2}$, we anticipate the attacker to be able to make a larger impact under the distributionally robust detector. In general, computing the exact (empirical) reachable set is computationally intensive, however, we plot each in this example for reference. 

When there is no attack, the traditional chi-squared detector can generate high rates of false alarms in worst case noise settings while the distributionally robust detector is guaranteed to remain below the designed false alarm rate. Fig. \ref{fig_worst_A} compares for each detector (chi-squared and distributionally robust) the false alarm rate they are tuned for (horizontal axis) with the worst possible false alarm rate they may generate under arbitrary noise distributions (vertical axis). It is evident that the worst case false alarm rate is much larger than the desired false alarm rate when a chi-squared detector is used. In contrast the distributionally robust detector's worst case false alarm rate is exactly what is designed.

However, this performance comes at the cost of increased attacker capabilities when the system is under attack. 
Fig. \ref{fig_vol_vs_alarm_rate} depicts the trade-off observed between the desired false alarm rate, $\mathcal{A}$, and the the trace of the ellipsoid that bounds the reachable set obtained using $\alpha_{DR}$ tuning. It is evident that as $\mathcal{A}$ increases, the trace of the distributionally robust bounding ellipsoid decreases. The same trend pertains even when the process noise covariance is varied, where the trace of the corresponding bounding ellipsoid is larger for the states driven by the process noise having higher covariance. This inversely proportional trend suggests it is increasingly more costly to tune the distributionally robust detector for smaller false alarm rates.

\section{Conclusion \& Future Outlook}\label{sec_conclusion} 
We have proposed a distributionally robust approach to tuning anomaly detectors by using the generalized Chebyshev moments to find a threshold that guarantees the false alarm rate will not exceed a desired value. We have demonstrated our ideas with a numerical example that emphasizes the effectiveness of the distributionally robust approach over the traditional chi-squared detector approach, however the advantages come at the price of increased attacker capabilities. Our future work will seek to ameliorate this downside as it is possible to obtain sharper probability estimates for the threshold values if we include the higher order moments as explained in \cite{bertsimas_popescu}.   

\bibliographystyle{IEEEtran}
\bibliography{security}

\end{document}